\documentclass[a4paper,11pt]{amsart}
\usepackage{amsmath,amsthm,amsfonts,amssymb,upgreek,multicol,mathrsfs,pifont,amscd,latexsym,cite,bm,geometry,color,fancyhdr,abstract,framed,appendix,mathtools}
\usepackage{graphicx,algorithm,algorithmic,tabularx,url,float,booktabs,array,threeparttable,tikz}
\usepackage{stmaryrd}

\newcommand{\bx}{\boldsymbol{x}}

\newcommand{\by}{\boldsymbol{y}}

\newcommand{\bv}{\boldsymbol{v}}
\newcommand{\bq}{\boldsymbol{q}}
\newcommand{\bs}{\boldsymbol{s}}
\newcommand{\bS}{\boldsymbol{S}}
\newcommand{\bw}{\boldsymbol{w}}
\newcommand{\bft}{\boldsymbol{f}} 

\newcommand{\curl}{\operatorname{curl}} 
\newcommand{\scurl}{ {\operatorname{curl}^{*} }}

\theoremstyle{plain}
\newtheorem{theorem}{Theorem}[section]

\newtheorem{df}[theorem]{Definition}
\newtheorem{lemma}[theorem]{Lemma}

\numberwithin{equation}{section}
\numberwithin{table}{section}
\numberwithin{figure}{section}

\begin{document}
	
	\title[]{Generalized vector potential and Trace Theorem for Lipschitz domains }
	\author {Zhen Liu}
	\address{LMAM and School of Mathematical Sciences, Peking University, Beijing 100871, P. R. China.\\ Chongqing Research Institute of Big Data, Peking University, Chongqing 401332, P. R. China. zliu37@pku.edu.cn}
	\author {Jinbiao Wu}
	\address{LMAM and School of Mathematical Sciences, Peking University, Beijing 100871, P. R. China.\\ Chongqing Research Institute of Big Data, Peking University, Chongqing 401332, P. R. China. jwu@math.pku.edu.cn}
	
	\maketitle
	
	\begin{abstract}
		The vector potential is a fundamental concept widely applied across various fields. This paper presents an existence theorem of a vector potential for divergence-free functions in $W^{m,p}(\mathbb{R}^N,\mathbb{T})$ with general $m,p,N$. Based on this theorem, one can establish the space decomposition theorem for functions in $W^{m,p}_0(\curl;\Omega,\mathbb{R}^N)$ and the trace theorem for functions in $W^{m,p}(\Omega)$ within the Lipschitz domain $\Omega \subset \mathbb{R}^N$. The methods of proof employed in this paper are straightforward, natural, and consistent. 
	\end{abstract}

	\section{Introduction}
	The vector potential plays a crucial role in physics and engineering, with applications spanning electromagnetic phenomena \cite{panofsky2005classical,griffiths2023introduction}, fluid dynamics \cite{batchelor1967introduction}, and quantum mechanics \cite{phillips2013introduction,aharonov1959significance}. Its versatility and importance arise from its ability to concisely represent intricate vector fields, thereby enabling the analysis of various physical systems.
	
	Given a divergence-free vector function $\bv$, it is cumbersome to construct a vector potential $\bw$ satisfying $\curl \bw = \bv$. Additionally, some researchers are concerned not only with the existence of the vector potential but also with its regularity. 
	The following theorem shows the existence of regular vector potentials for some divergence-free Sobolev functions with $p=2$ in $\mathbb{R}^3$. {There are many analyses and proofs of theorems related to vector potential based on this result, see \cite{hiptmair2002finite,amrouche1998vector}.} 
	\begin{theorem}
		\label{thmregular1}
		For any integer $m \geq 0$, there exists a continuous mapping
		$$ \mathrm{L} :  {H}(\operatorname{div}  0;\mathbb{R}^3, \mathbb{R}^3) \cap {H}^{m}(\mathbb{R}^3, \mathbb{R}^{3}) \rightarrow {H}^{m+1}_{\mathrm{loc}}(\mathbb{R}^{3}, \mathbb{R}^3),$$
		such that $\operatorname{curl}  \mathrm{L} \bv = \boldsymbol{v}$ and $\operatorname{div} \mathrm{L} \bv= 0$.    
	\end{theorem}
	This paper mainly extends Theorem \ref{thmregular1} to Theorem \ref{thmmain}, demonstrating the existence of a vector potential for functions in $W^{m,p}(\mathbb{R}^N,\mathbb{T})$ with $m \geq 0, 1< p < \infty$, and $N\geq 2$.
	This generalization is not straightforward. {One challenge} arises from the complexities introduced in defining operators in any $N$ dimensions. 
	For instance, the action of the $\curl$ operator on the vector will yield an anti-symmetric tensor, as defined in Definition \ref{curlh}. {Another challenge is that certain properties of Hilbert spaces when $p=2$ may not extend to Sobolev spaces with general values of $p$.} 
	
	Based on the results in Theorem \ref{thmmain}, a decomposition result of ${W}_{0}^{m,p}(\operatorname{curl}; \Omega,\mathbb{R}^N)$ in Theorem \ref{thmdecompostion} can be proved.  An essential step involves in the proof extending the function defined on $\Omega$ to $\mathbb{R}^N$ and then utilizing Theorem \ref{thmmain} to determine a vector potential. By following similarly procedures, this paper gives decomposition results for $L^{p}(\operatorname{curl}; \Omega,\mathbb{R}^N)$ and $W^{1,p}(\operatorname{curl}; \Omega,\mathbb{R}^3)$.
	\begin{theorem}
		\label{thmdecompostion}
		The following decomposition holds
		$$
		{W}_{0}^{m,p}(\operatorname{curl}; \Omega,\mathbb{R}^N)  = {W}_{0}^{m+1,p}(\Omega,\mathbb{R}^N) + \nabla(W^{m+1,p}_{0}(\Omega)), \quad m \geq 1.
		$$
		where $\Omega \subset \mathbb{R}^N$ is a {bounded Lipschitz} domain.
	\end{theorem}

	Theorem \ref{thmdecompostion} not only clarifies the structure of ${W}_{0}^{m,p}(\operatorname{curl}; \Omega,\mathbb{R}^N)$ but also facilitates the proof of the trace theorem. Trace theory is instrumental in establishing the existence of weak solutions and regularity results in partial differential equations. However, most existing results focus on domains with smooth boundaries or domains that are polyhedral \cite{lamberti2020trace, gagliardo1957caratterizzazioni,burenkov1998sobolev}. This paper primarily addresses the trace space of $W^{m,p}(\Omega)$ for arbitrary $m$ and $N$ on Lipschitz domain $\Omega \subset \mathbb{R}^N$ in Theorem \ref{thmgeneraltrace}. In \cite[Theorem 7.8 and Corollary 7.11]{maz2010dirichlet}, this has been proved by using an analytical method based on Taylor expansions in Besov and weighted Sobolev spaces. This paper provides a new and unified proof. 
	
	The organization of this paper is as follows: Section 2 introduces notations, definitions, and preliminary lemmas. Section 3 details the construction of the generalized potential function and provides a rigorous proof for the simpler case outlined in Theorem \ref{thmdecompostion}. Section 4 examines the necessary and sufficient conditions for the trace theorem applicable to functions in $W^{m,p}(\Omega)$ and includes the proof for the specific case when $m=3$.

	\section{Notations And Preliminaries}
	
	In this section, the necessary notations and preliminary concepts are introduced to facilitate later discussions and analyses.
	
	\subsection{Notations}
	In this paper, unless otherwise specified, lowercase English letters and Greek letters denote scalars and scalar functions, respectively. Bold lowercase English letters denote vectors or vector functions, while bold uppercase English letters denote tensors or tensor functions of at least second order.
	
	For $1 \leq p<\infty$ and $m$ integer $W^{m, p}(D,X)$ denotes the Sobolev space of functions within domain $D$, taking values in space $X$, of $L^{p}(D,X)$ whose distributional derivatives up to the order $m$ also belong to $L^{p}(D,X)$. {Clearly, $W^{0,p}(D,X) = L^{p}(D,X)$.}
	{The corresponding norm and semi-norm are denoted as $\| \cdot \|_{W^{m, p}(D,X)}$ and $| \cdot |_{W^{m, p}(D,X)}$, respectively.}
	Similarly, let $C^m(D , X)$ denote the space of $m$-times continuously differentiable functions, taking values in $X$. {Let $C^{\infty}(D,X) = \bigcap_{m=0}^{\infty} C^{m}(D,X)$ and $C^{\infty}_0(D,X)$ consists of all those functions in $C^{\infty}(D,X)$ that has compact support in $D$.} Let {$W^{m,p}_0(D,X)$ be the closure of $C_{0}^{\infty}(D,X)$ in the space $W^{m, p}(D,X)$.}
	In this paper, $X$ could be $\mathbb{R}, \mathbb{R}^{N}$ or $\mathbb{T}$, where $\mathbb{T}$ denotes the set of all skew-symmetry $\mathbb{R}^{N\times N}$ matrices. If $X=\mathbb{R}$, then $W^{m,p}(D)$ abbreviates $W^{m,p}(D,X)$, similarly for $C^m(D)$. When $p=2$, $W^{m,2}(D,X)$ becomes Hilbert spaces and denoted as $H^{m}(D,X)$. Define 
	\begin{equation*}
		\label{defHspace}
		\begin{aligned}
			H(\operatorname{div}; \mathbb{R}^N, \mathbb{R}^N) & \coloneqq \{  \boldsymbol{v} \in L^2(\mathbb{R}^N,\mathbb{R}^N); \operatorname{div} \boldsymbol{v} \in L^2(\mathbb{R}^N) \}, \\
			H(\operatorname{div~0}; \mathbb{R}^N,\mathbb{R}^N ) &\coloneqq \{  \boldsymbol{v} \in H(\operatorname{div}; \mathbb{R}^N,\mathbb{R}^N); \operatorname{div} \boldsymbol{v} =0 \},
		\end{aligned}
	\end{equation*}
	which are used in Theorem \ref{thmregular1}.

	In this paper, let $\Omega \subset \mathbb{R}^N$ denote a {Lipschitz connected domain} as defined in Definition \ref{dfLipschitz}, with an {oriented compact} boundary denoted as $\Gamma$.
	
	\begin{df}[Lipschitz Domain] 
		\label{dfLipschitz}
		If for every $\bx \in \Gamma$, there exists a neighborhood $V(\bx)=\prod_{i=1}^{N} [a_{i}, b_{i} ] =V^{\prime}(\bx) \times [a_{N}, b_{N} ]$ of $\bx$ in $\mathbb{R}^{N}$ and a Lipschitz continuous function $\varphi: V^{\prime}(\bx) \rightarrow [ a_{N}, b_{N} ]$ such that\\
		(1) $\left|\varphi\left(y^{\prime}\right)\right| \leq \operatorname{max}(|a_N|,|b_N|) / 2$ for every $y^{\prime}=$ $\left(y_{1}, \ldots, y_{N-1}\right) \in V^{\prime}(\bx)$,\\
		(2) $\Omega \cap V(\bx)=\left\{\by=\left(y^{\prime}, y_{N}\right) \in V(\bx) \mid y_{N}\right.$ $\left.>\varphi\left(y^{\prime}\right)\right\}$, \\
		(3) $\Gamma \cap V(\bx)=\left\{\by=\left(y^{\prime}, y_{N}\right) \in V(\bx) \mid y_{N}=\varphi\left(y^{\prime}\right)\right\}$.\\
		Then the domain $\Omega \subset \mathbb{R}^N$ is called a Lipschitz domain.
	\end{df}
	
	There exist $s$ open and connected subsets $\Gamma_{i}$ such that $\Gamma=\bigcup_{i=1}^{s} \Gamma_{i}$. Within each subset $\Gamma_{i}$, there exist a point $o_{i} \in \Gamma_{i}$ and a Lipschitz continuous function $\varphi_{i}$ such that $\Gamma_{i}=\Gamma \cap V\left(o_{i}\right)$. The induced parametrization $\left(\varphi_{i}, \Gamma_{i}\right)$ of $\Gamma$ is defined for $i=1, \ldots, s$ by
	$y^{\prime}=\left(y_{1}, \ldots, y_{N-1}\right) \rightarrow \left(y^{\prime}, \varphi_{i}\left(y^{\prime}\right)\right)$. The unit outward normal vector $\boldsymbol{n} \in L^{\infty}(\Gamma,\mathbb{R}^N)$ and the {unit tangent vectors} $\boldsymbol{\tau}_{k} \in T_{x} \Gamma, k=1, \ldots, N-1$, can be defined a.e. on $\Gamma_{i}$ where $T_{x} \Gamma$ is the tangent space of the point $\bx \in \Gamma_i$. Therefore, the vectors $\left(\boldsymbol{\tau}_{1}, \ldots, \boldsymbol{\tau}_{N-1}, \boldsymbol{n}\right)$ are a positively oriented basis of $T_{x} \Omega$ for a.e. $\bx \in \Gamma$. 
	
	\subsection{Preliminaries} {This subsection extends the definition of the $\curl$ operator to higher dimensions and provides several identities and estimates of Newton potential functions for subsequent use.}
	
	The following definitions and lemmas are adapted from \cite{Oliver2023, arnold2018finite}.  
	
	\begin{df}[$\curl$ operator in $N$ dimensions]
		\label{curlh}
		Define operator $\curl: C^1\left(\Omega, \mathbb{R}^N\right) \rightarrow$ $C^0\left(\Omega, \mathbb{T}\right)$ for a vector function $\bft$ as anti-symmetric matrix function $\boldsymbol{A}$ with i-j-element defined as follows
		$$
		{A}_{ij} =[ \curl \bft]_{ij}:=\frac{1}{2} \left(\frac{\partial f_j}{\partial x_i}-\frac{\partial f_i}{\partial x_j} \right), \quad 1 \leq i, j \leq N. 
		$$
	\end{df}
	
	\begin{df}[$\scurl$ operator in $N$ dimensions]
		\label{surlh}
		Define operator $\scurl: C^1\left(\Omega, \mathbb{T}\right) \rightarrow$ $C^0\left(\Omega, \mathbb{R}^N \right)$ of an anti-symmetric matrix function $\boldsymbol{A}$ as vector function $\bft$ with i-element defined as follows
		$$
		f_i = [ \scurl \boldsymbol{A}]_{i}:= \sum_{j=1}^N 2 \frac{\partial A_{ij}}{\partial x_j}, \quad 1 \leq i \leq N. 
		$$
	\end{df}
	
	Unless otherwise specified, the Einstein summation convention is used in this paper.
	Define the trace operator $\gamma_t (\bft)$ on $\Gamma$ with $i$-$j$-element as $\frac{1}{2}(f_j n_i - f_i n_j)$. {Following Definition \ref{curlh}, introduce the spaces}
	\begin{eqnarray}
		\label{defspace}
		\begin{aligned}
			{L}^{p}(\operatorname{curl}; \Omega,\mathbb{R}^N) &\coloneqq \{ \bv \in {L}^{p}(\Omega,\mathbb{R}^N);  \curl \bv \in L^p(\Omega, \mathbb{T}) \}, \\
			{L}_{0}^{p}(\operatorname{curl}; \Omega,\mathbb{R}^N) &\coloneqq \{ \bv \in {L}^{p}(\curl; \Omega,\mathbb{R}^N);  \gamma_t(\bv)= 0  \text{ on } \Gamma \}, \\
			{W}_{0}^{m,p}(\operatorname{curl}; \Omega,\mathbb{R}^N) &\coloneqq \{ \bv \in {W}_{0}^{m,p}(\Omega,\mathbb{R}^N);  \operatorname{curl} \bv \in {W}_{0}^{m,p}(\Omega,\mathbb{T})  \} (m \geq 1).
		\end{aligned}
	\end{eqnarray}
	
	Given two vector functions $\boldsymbol{f}$ and $\boldsymbol{g}$, let
	\begin{equation*}
		\begin{aligned}
			(\boldsymbol{f}, \boldsymbol{g})_{\Omega} \coloneqq \int_{\Omega} \boldsymbol{f} \cdot \boldsymbol{g} \mathrm{d} \bx.
		\end{aligned}
	\end{equation*}
	Given two tensor functions $\boldsymbol{A}$ and $\boldsymbol{B}$, in order to avoid unnecessary coefficient in the following derivation, specially, let
	\begin{equation}
		\label{definner}
		(\boldsymbol{A},\boldsymbol{B})_{\Omega} \coloneqq 2\int_{\Omega} A_{ij}B_{ij} \mathrm{d} \bx.
	\end{equation}
	The subscript $\Omega$ will be omitted when there is no ambiguity.
	
	\begin{lemma}
		\label{lemmacurlcurh}
		The following equality holds
		$$
		(\curl \boldsymbol{g}, \boldsymbol{A}) = (\boldsymbol{g}, \scurl \boldsymbol{A}), \quad \forall~ \boldsymbol{g} \in C^1\left(\Omega, \mathbb{R}^N \right), \boldsymbol{A} \in C^{\infty}_0\left(\Omega, \mathbb{T}\right).
		$$
		
	\end{lemma}
	\begin{proof}
		{By Definition \ref{curlh} and \eqref{definner}, one can derive that}
		$$
		\begin{aligned}
			(\curl \boldsymbol{g}, \boldsymbol{A}) &= 2\int_{\Omega}\frac{1}{2}  \left(\frac{\partial g_j}{\partial x_i}-\frac{\partial g_i}{\partial x_j} \right) A_{ij} \mathrm{d} \bx \\
			& = - \left( g_j ,\frac{\partial A_{ij}}{\partial x_i} \right) + \left( g_i, \frac{\partial A_{ij}}{\partial x_j} \right) \\
			& = -\left( g_i ,\frac{\partial A_{ji}}{\partial x_j} \right) + \left( g_i, \frac{\partial A_{ij}}{\partial x_j} \right) = 2\left( g_i , \frac{\partial A_{ij}}{\partial x_j} \right).
		\end{aligned}
		$$
		The proof is completed by noting Definition \ref{surlh}.
	\end{proof}

	\begin{lemma}
		\label{lemmalaplace1}
		For any twice continuously differentiable vector field $\bft \in C^2\left(\Omega, \mathbb{R}^N \right)$, the following identity holds
		\begin{equation*}
			-\Delta \bft = -\operatorname{grad} \operatorname{div} \bft+ \scurl \curl \bft.
		\end{equation*}
	\end{lemma}
	\begin{proof}
		{It follows from Definition \ref{curlh} and \ref{surlh} that } 
		$$
		\begin{aligned}
			&\quad [ -\operatorname{grad} \operatorname{div} \bft + \scurl \curl \bft  ]_i \\
			& = -\frac{\partial}{\partial x_i}  \frac{\partial f_j}{\partial x_j}+ 2\cdot \frac{1}{2} \frac{\partial}{\partial x_j}\left(\frac{\partial f_j}{\partial x_i}-\frac{\partial f_i}{\partial x_j}\right) \\
			& = - \frac{\partial}{\partial x_j}\frac{\partial f_i}{\partial x_j}.
		\end{aligned}
		$$
		{The cancellation of the first two terms occurs due to the symmetry of second derivatives and the interchange of the sum and derivative.}
	\end{proof}
	
	{This paper will frequently employ solutions of the Laplace equation to construct potential functions satisfying the requirements. Consequently, the definition of the kernel function in $N$-dimensional space is presented.}
	
	\begin{df}[Kernel function in $N$ dimensions] 
		\label{kernel}	
		Define the fundamental solutions of the Laplace equation as follows 
		\begin{equation}
			\label{defkernel}
			\lambda(\bx-\by)= \begin{cases}-\frac{1}{2 \pi}(\log |\bx-\by|) & N=2, \\ \frac{1}{N(N-2) V_N}\left(|\bx-\by|^{2-N}\right) & \text { otherwise, }\end{cases}
		\end{equation}
		where $V_N$ denotes the volume of a unit $N$-ball.
	\end{df}
	
	For function $\phi(\bx)$ with compact support, define the Newton potential function using the kernel function in \eqref{defkernel} as
	\begin{equation*}
		\varphi(\boldsymbol{x})=\int_{\mathbb{R}^N} \phi(\boldsymbol{y}) \lambda(\boldsymbol{x}- \boldsymbol{y}) d \boldsymbol{y}.
	\end{equation*}
	It is easy to verify that $-\triangle \varphi(\bx) = \phi(\bx)$ and the following estimates hold.
	
	\begin{lemma}[{\cite[Theorem 9.11]{gilbarg1977elliptic}}]
		\label{InteriorEstimate}
		Let $\tilde{\Omega}$ be an open set in $\mathbb{R}^N$ and $\psi \in W^{2,p}_{\mathrm{loc}}(\tilde{\Omega}) \cap L^p(\tilde{\Omega}),1<p<\infty$, be a strong solution of the equation $-\Delta \varphi = \phi$ in $\tilde{\Omega}$. Then for any domain $\Omega \subset \subset \tilde{\Omega}$, where $\Omega \subset \subset \tilde{\Omega}$ if the closure of $\Omega$ denoted as $\bar{\Omega} \subset \tilde{\Omega}$ and $\bar{\Omega}$ is a compact subset of $\mathbb{R}^N$, it holds that
		$$
		\| \varphi \|_{W^{2,p}(\Omega)} \leq  C(N,p,\Omega,\tilde{\Omega}) \left(\| \varphi \|_{L^p(\tilde{\Omega})} + \| \phi \|_{L^p(\tilde{\Omega})} \right).
		$$
	\end{lemma}
	
	\begin{lemma}[Calderon-Zygmund inequality\cite{Calderon1952}]
		\label{CZ}
		Let $\phi \in L^p(\mathbb{R}^N), 1<p<\infty$, and let $\varphi$ be the Newtonian potential of $\phi$. Then there exists a constant $C(N,p)$ such that $$\quad\left\| \partial_{ij}^2 \varphi \right\|_{L^p({\mathbb{R}^N})} \leq C(N,p) \|\phi \|_{L^p({\mathbb{R}^N})}, \quad i,j=1,2,\cdots,N.$$
	\end{lemma}
	
	\begin{lemma}
		\label{Wu}
		{Let }$\phi \in L^p(\mathbb{R}^N), 1<p<\infty$, and let $\varphi$ be the Newtonian potential of $\phi$. Then {for any bounded domain $\Omega$}, there exists a constant $C(N,p,\Omega)$ such that $$\quad\left\| \varphi \right\|_{W^{1,p}(\Omega)} \leq C(N,p,\Omega) \|\phi \|_{L^p({\mathbb{R}^N})}.$$
	\end{lemma}

	\section{Generalized Vector Potential}
	\label{sec3}
	
	This section first offers a detailed proof of the generalized vector potential theorem. Subsequently, it presents several decomposition theorems derived from this principal theorem. The subsequent assumptions and theorems are are necessary for the proofs presented later.
	
	\newenvironment{assumption}[1]{\medskip\noindent{\textbf{Assumption} (#1):}\rmfamily\label{#1}}{\par\medskip}
	\newcommand{\cA}[1]{(\ref{#1})}
	\begin{assumption}{A1}
		In this paper, it is consistently assumed that if $\curl \bv =0$ in $\Omega$, then there exists a function $\eta$ such that $\bv = \nabla \eta$.
	\end{assumption}

	\begin{theorem}[Weyl's lemma{\cite[corollary 1.2.1]{jost2012partial}}]
		\label{weyl}
		Let $\varphi: \Omega \rightarrow \mathbb{R}$ be measurable and locally integrable in $\Omega$. Suppose that for all $\psi \in C_0^{\infty}(\Omega)$,
		$$
		\int_{\Omega} \varphi(x) \Delta \psi(x) d x=0 .
		$$
		Then $\varphi$ is harmonic and, in particular, smooth.
	\end{theorem}
	
	\begin{theorem}[The Stein extension theorem {\cite[Theorem 5.24]{adams2003sobolev}}] 
		\label{thmextension}
		Let $\Omega \subset \mathbb{R}^N$ be a bounded Lipschitz domain and $m \geq 1$ be an integer, $1 \leq p \leq \infty$. Then $\Omega$ has a bounded extension operator $E: W^{m,p}(\Omega) \rightarrow W^{m,p}(\mathbb{R}^N)$,i.e., for $\phi \in W^{m,p}(\Omega)$, there exists a constant $C$ such that 
		$$
		E\phi(\bx)=\phi(\bx) \text{ a.e. in } \Omega, \quad \text{and} \quad \| E\phi \|_{W^{m,p}(\mathbb{R}^N)} \leq C(m,p) \| \phi\|_{W^{m,p}(\Omega)}.
		$$
	\end{theorem}

	It is now appropriate to present the main theorem.
	
	\begin{theorem}[{Generalized vector potential theorem}]
		\label{thmmain}
		Assume that vector function $\boldsymbol{v}$ has compact support set in $\mathbb{R}^{N}$, $\boldsymbol{v} \in {L}^{p}(\mathbb{R}^{N},\mathbb{R}^{N})$, and $\operatorname{curl}  \boldsymbol{v} \in {W}^{m,p}(\mathbb{R}^{N},\mathbb{T})$, then there exists $\bw \in {W}_{\mathrm{loc}}^{m+1,p}(\mathbb{R}^{N},\mathbb{R}^{N})$ such that 
		$\operatorname{curl}  \bw=  \operatorname{curl} \boldsymbol{v}$ and $\operatorname{div}\bw= 0$.
		Additionally, the following estimate holds for every bounded domain $\Omega$,
		$$
		\| \bw \|_{W^{m+1,p}(\Omega,\mathbb{R}^N)} \leq C(N,p,\Omega)\left(\| \bv \|_{L^{p}(\mathbb{R}^N, \mathbb{R}^N)} + \| \curl \bv  \|_{W^{m,p}(\mathbb{R}^N,\mathbb{T})}\right).
		$$
	\end{theorem}
	\begin{proof}
		{The results of the above theorem will be demonstrated constructively. The construction approach remains generally consistent across various values of $m$. However, for $m=0$, the overall lower regularity complicates the construction process. Hence, a detailed description is provided specifically for the case of $m=0$.} \\
		{\textbf{Step 1} constructs $\bw \in   {L}^{p}(\mathbb{R}^N, \mathbb{R}^N)$ } such that $\operatorname{curl} \bw = \operatorname{curl} \boldsymbol{v}$ and $\operatorname{div} \boldsymbol{w} = 0$. Define 
		\begin{equation}
			\label{g}
			\eta(\bx) = \operatorname{div}\int_{\mathbb{R}^{N}} \lambda(\bx-\by) \boldsymbol{v}(\by)d\by =  \operatorname{div} \boldsymbol{g}(\bx).
		\end{equation}
		Let $\bw = \bv + \nabla \eta = \bv + \nabla \operatorname{div} \boldsymbol{g}$. {It follows by Lemma \ref{CZ} that}
		\begin{equation}
			\label{estimate1}
			\begin{aligned}
				\| \bw\|_{L^p(\mathbb{R}^N,\mathbb{R}^N)} 
				\leq  \| \bv\|_{L^p(\mathbb{R}^N, \mathbb{R}^N)} + C(N,p)\| \bv\|_{L^p(\mathbb{R}^N,\mathbb{R}^N)} 
				\leq C(N,p) \| \bv\|_{L^p(\mathbb{R}^N,\mathbb{R}^N)}.
			\end{aligned}
		\end{equation}
		Thus $\bw \in {L}^{p}(\mathbb{R}^N, \mathbb{R}^N)$. It is evident that $\operatorname{curl} \boldsymbol{\omega} = \operatorname{curl} \boldsymbol{v}$ because of $\curl ( \nabla \operatorname{div} \boldsymbol{g} ) =0$. By the definition of $\boldsymbol{g}(\bx)$ in \eqref{g}, it follows that $-\triangle {g}_i = {v}_i, i =1,2,\cdots,N$. Thus,
		\begin{equation}
			\label{eqnotsum1}
			(\nabla g_i,  \frac{\partial }{\partial x_i}\nabla \psi ) =  (\nabla g_i, \nabla \frac{\partial }{\partial x_i} \psi) = (v_i, \frac{\partial \psi}{\partial x_i}),\quad  i =1,2,\cdots,N,\quad  \forall \psi \in C_0^{\infty}(\mathbb{R}^N),
		\end{equation}
		which implies 
		\begin{equation}
			\label{eqnotsum2}
			-(\nabla  \frac{\partial }{\partial x_i} g_i,  \nabla \psi) = (v_i, \frac{\partial \psi}{\partial x_i}),\quad i =1,2,\cdots,N,\quad  \forall \psi \in C_0^{\infty}(\mathbb{R}^N).
		\end{equation}
		It is important to note that the notation {used in \eqref{eqnotsum1} and \eqref{eqnotsum2} do not} employ Einstein summation convention.
		Summing up \eqref{eqnotsum2} for $i =1,2,\cdots,N$, it follows that
		$$
		-(\nabla (\operatorname{div} \boldsymbol{g(x)}), \nabla \psi) = (v, \nabla \psi), \quad \forall \psi \in C_0^{\infty}(\mathbb{R}^N), 
		$$
		which proves that $\operatorname{div} \boldsymbol{w}= 0$.\\	 
		{{\textbf{Step 2} constructs $\boldsymbol{w}_1 \in {W}_{\mathrm{loc}}^{1,p}(\mathbb{R}^{N},\mathbb{R}^N)$}} satisfies 
		\begin{equation}
			\label{comb}
			(\nabla \boldsymbol{w}_1, \nabla \bq ) = (\operatorname{curl} \boldsymbol{v}, \operatorname{curl} \bq),\quad  \forall \bq \in C^{\infty}_0(\mathbb{R}^N, \mathbb{R}^N).
		\end{equation}
		Let $$\boldsymbol{w}_1 = \operatorname{curl}^{*}\int_{\mathbb{R}^{N}} \lambda(\bx- \by) \operatorname{curl}\boldsymbol{v}(\by)d\by = \operatorname{curl}^{*}\boldsymbol{H}(\bx).$$
		{Then for any bounded domain $\Omega$, it follows by Lemma \ref{CZ} and \ref{Wu} that}
		\begin{equation*}
			\begin{aligned}
				\| \bw_{1} \|_{W^{1,p}(\Omega,\mathbb{R}^N)} &= \| \bw_{1} \|_{L^{p}(\Omega,\mathbb{R}^N)} + | \bw_{1} |_{W^{1,p}(\Omega,\mathbb{R}^N)} \\
				& \leq  | \boldsymbol{H} |_{W^{1,p}(\Omega,\mathbb{T})} + | \boldsymbol{H} |_{W^{2,p}(\Omega,\mathbb{T})}   \leq C(N,p,\Omega) \|\curl \bv \|_{L^p(\mathbb{R}^N,\mathbb{R}^N)}.
			\end{aligned}
		\end{equation*}
		Thus $\bw_1 \in {W}_{\mathrm{loc}}^{1,p}(\mathbb{R}^{N},\mathbb{R}^N)$. {Following similar procedures as Step 1, one can verify that $- \triangle \bw_1 = \operatorname{curl}^{*} \operatorname{curl} \bv$ in weak sense, i.e., 
			\begin{equation*}
				\label{comb1}
				(\nabla \boldsymbol{w}_1, \nabla \bq ) = (\nabla \operatorname{curl}^{*} \boldsymbol{H}, \nabla \bq ) = (\operatorname{div} \operatorname{\nabla} \boldsymbol{H}, \operatorname{curl} \bq ) = (\operatorname{curl} \boldsymbol{v}, \operatorname{curl} \bq),\forall \bq \in {C}^{\infty}_0(\mathbb{R}^{N},\mathbb{R}^N).
		\end{equation*}}
		{{\textbf{Step 3} proves $\boldsymbol{w}-\bw_1 \in    {W}_{\mathrm{loc}}^{1,p}(\mathbb{R}^N, \mathbb{R}^N)$}}.
		{It follows from Step1 that }
		\begin{equation}
			\label{comb2}
			(\operatorname{curl} \boldsymbol{w}, \operatorname{curl} \bq )_{\mathbb{T}} + (\operatorname{div} \boldsymbol{w}, \operatorname{div} \bq) = (\operatorname{curl} \boldsymbol{v}, \operatorname{curl} \bq ), \quad \forall \bq \in {C}^{\infty}_0(\mathbb{R}^N, \mathbb{R}^N).
		\end{equation}
		{According to Lemma \ref{lemmacurlcurh} and \ref{lemmalaplace1}, the following identity holds}
		\begin{equation}
			\label{comb3}
			(\boldsymbol{w}, \triangle \bq ) =  (\operatorname{curl} \boldsymbol{w}, \operatorname{curl} \bq) + (\operatorname{div} \boldsymbol{w}, \operatorname{div} \bq), \quad \forall \bq \in {C}^{\infty}_0(\mathbb{R}^N, \mathbb{R}^N).
		\end{equation}
		Combining \eqref{comb2} \eqref{comb3} and \eqref{comb} together, one can obtain
		$$(\boldsymbol{w} - \boldsymbol{w}_1, \triangle \bq )=0,\quad \forall \bq \in {C}^{\infty}_0(\mathbb{R}^N, \mathbb{R}^N).$$
		Let $\bw_2 = \bw - \boldsymbol{w}_1$. It follows by Weyl's lemma (Theorem \ref{weyl}) that $\bw_{2} \in {W}_{\mathrm{loc}}^{1,p}(\mathbb{R}^{N},\mathbb{R}^N)$. Hence, $\boldsymbol{w} = \bw_1 + \bw_2 \in {W}_{\mathrm{loc}}^{1,p}(\mathbb{R}^{N},\mathbb{R}^N)$.  \\
		{{\textbf{Step 4} estimates $\| \bw \|_{W^{1,p}(\Omega,\mathbb{R}^N)}$ for any compact set $\Omega \subset \mathbb{R}^N$}}. By \eqref{estimate1}, {it holds that}
		\begin{equation*}
			\begin{aligned}
				\| \bw\|_{L^p(\Omega,\mathbb{R}^N)} \leq \| \bw\|_{L^p(\mathbb{R}^N,\mathbb{R}^N)} 
				\leq C(N,p) \| \bv\|_{L^p(\mathbb{R}^N,\mathbb{R}^N)}.
			\end{aligned}
		\end{equation*}
		{There exists an open set $\tilde{\Omega}$ such that $\Omega \subset \subset \tilde{\Omega}$ due to $\Omega$ is bounded. Owing to $-\triangle \bw_2 = 0$, the function $\bw_2$ has the smoothness enough for estimating. By Lemma \ref{InteriorEstimate},}
		\begin{equation}
			\label{estimatew2}
			\begin{aligned}
				|\bw_2|_{W^{1,p}(\Omega, \mathbb{R}^N)} \leq \|\bw_2\|_{W^{2,p}(\tilde{\Omega}, \mathbb{R}^N)} \leq C(N,p,\Omega,\tilde{\Omega})\| \bw_2\|_{L^p(\tilde{\Omega}, \mathbb{R}^N)},
			\end{aligned}
		\end{equation}
		Note that $\tilde{\Omega}$ depends on $\Omega$ by construction. {Then the following estimate holds} 
		\begin{equation}
			\label{estimatew}
			\begin{aligned}
				|\bw|_{W^{1,p}(\Omega, \mathbb{R}^N)} &\leq  |\bw_1 |_{W^{1,p}(\Omega,\mathbb{R}^N)} +  |\bw_2|_{W^{1,p}(\Omega, \mathbb{R}^N)} \\
				&\leq  C(N,p)\| \operatorname{curl} \bv \|_{L^p(\mathbb{R}^N, \mathbb{R}^N)} + C(N,p,\Omega)\| \bw_2\|_{L^p(\tilde{\Omega}, \mathbb{R}^N)} \\
				& \leq C(N,p,\Omega)\left(  \| \operatorname{curl} \bv \|_{L^p(\mathbb{R}^N, \mathbb{R}^N)} + \| \bw_1 \|_{L^p(\tilde{\Omega}, \mathbb{R}^N)} + \| \bw \|_{L^p(\tilde{\Omega}, \mathbb{R}^N)} \right).
			\end{aligned}
		\end{equation}
		Due to $-\triangle \boldsymbol{H}(\bx) = \curl \bv$, it follows from Lemma \ref{Wu} that
		\begin{equation}
			\label{estimate3}
			\| \bw_1 \|_{L^p(\tilde{\Omega}, \mathbb{R}^N)} \leq | \boldsymbol{H}(x) |_{W^{1,p}(\tilde{\Omega}, \mathbb{R}^N)} \leq C(N,p,\Omega) \| \operatorname{curl} \bv \|_{L^p(\mathbb{R}^N,\mathbb{R}^N)}.
		\end{equation}
		The estimate for $m=0$ is verified by combining \eqref{estimatew2}, \eqref{estimatew}, and \eqref{estimate3}. 
		
		The proof is completed for $m=0$ by combining the above steps. For the case $m \geq 1$, due to the improvement in regularity, it is possible to appropriately enhance the aforementioned part of the proof process. Specifically, one can define
		$$ \eta(\bx) = \int_{\mathbb{R}^{N}} \lambda(\bx-\by)\operatorname{div} \boldsymbol{v}(\by) \mathrm{d} \by.$$And it follows from $\curl \bv \in W^{1,p}(\mathbb{R}^N,\mathbb{T})$ that
		$$ -\triangle \bw_{1} = \operatorname{curl}^{*} \operatorname{curl} \boldsymbol{v} \quad \text{in} \quad \mathbb{R}^{N}.$$
	Utilizing a similar approach allows for the completion of the remaining proof.
	\end{proof}
	
	Employing the above generalized vector potential theorem, one can derive the general decomposition result in Theorem \ref{thmdecompostion}. For clarity, the details of the case $m=1$ are provided therein. 
	
	\begin{theorem}
		\label{vd}
		The following decomposition holds
		$$
		{W}_{0}^{1,p}(\operatorname{curl}; \Omega,\mathbb{R}^N)  = {W}_{0}^{2,p}(\Omega,\mathbb{R}^N) + \nabla(W^{2,p}_{0}(\Omega)),
		$$
		where $\Omega \subset \mathbb{R}^N$ is a bounded Lipschitz domain.
	\end{theorem}

	\begin{proof}Given $\bv \in{W}_{0}^{1,p}(\operatorname{curl}; \Omega,\mathbb{R}^N)$ defined in \eqref{defspace}, then $\boldsymbol{v}$ satisfies
		$$
		\gamma_{0} (\boldsymbol v )= 0, \quad \gamma_{0}( \operatorname{curl} \boldsymbol{v}) =  {0}, \text{ on } \Gamma.
		$$
		Let $\tilde{\boldsymbol{v}}$ be the zero extension of vector function $\boldsymbol{ v}$ from $\Omega$ to $\mathbb{R}^N$, i.e., 
		\begin{equation}
			\begin{aligned}
				\label{eqtbv}
				\tilde{\boldsymbol{ v}} = \boldsymbol{ v}, \text{ in }  \Omega,\quad  \tilde{\boldsymbol{ v}} = 0, \text{ in }   \Omega^{c}.
			\end{aligned}
		\end{equation}
		Then $\tilde{\bv} \in  {W}^{1,p}(\mathbb{R}^{N},\mathbb{R}^N)$ and $\operatorname{curl} \tilde{\bv} \in {W}^{1,p}(\mathbb{R}^{N},\mathbb{R}^N)$. By Theorem \ref{thmmain}, there exists a function $\tilde{\bw} \in {W}_{\mathrm{loc}}^{2,p}(\mathbb{R}^{N},\mathbb{R}^N) $ such that
		\begin{equation}
			\label{eqtbw}
			\operatorname{curl}  \tilde{\bw} = \operatorname{curl} \tilde{\bv}, \quad \operatorname{div} \boldsymbol{\tilde \omega} = 0.
		\end{equation}
		It follows from \eqref{eqtbw} that $\operatorname{curl} ({\tilde \bv} - {\tilde \bw}) = 0$ in $\mathbb{R}^N$. According to Assumption (A1), there exists a function $\eta \in W_{\mathrm{loc}}^{2,p}(\mathbb{R}^{N})$ such that $\nabla \eta = {\tilde \bv} - {\tilde \bw}$. By \eqref{eqtbv}, it follows that
		\begin{equation}
			\label{vd3}
			\begin{aligned}
				\boldsymbol{ v}  &= { \tilde{\bw} } + \nabla \eta,  \quad  \text{in } \Omega,  \\
				{0}  &= \tilde{\bw}  + \nabla \eta,  \quad  \text{in }  \Omega^{c}. \\
			\end{aligned}
		\end{equation}
		{Due to the restriction of} $\tilde \bw$ to ${\Omega^{c}}$ belongs to ${W}_{\mathrm{loc}}^{2,p}(\Omega^{c},\mathbb{R}^N)$ in \eqref{vd3}, then $ \eta \in {W}_{\mathrm{loc}}^{3,p}(\Omega^{c})$.  
		{It follows by Theorem \ref{thmextension} that there exists a $\tilde \eta \in W^{3,p}_{\mathrm{loc}}(\mathbb{R}^{N})$ which is the extension of $\eta$ from ${\Omega^{c}}$ to $\mathbb{R}^{N}$. Note that $\eta$ may not equal to $\tilde \eta$ in $\Omega$ though $\eta = \tilde \eta$ in $\Omega^{c}$.} However, the following relationship holds
		\begin{equation}
			\label{eqdecomposition}
			\begin{aligned}
				\boldsymbol{v}  = {\tilde{\bw}} + \nabla \eta  = {\tilde{\bw}} + \nabla{\tilde{\eta}} + \nabla(\eta - \tilde{\eta}), \text{ in } \Omega.
			\end{aligned}
		\end{equation}
		Since ${\tilde{\bw}} + \nabla{\tilde \eta} \in  {W}_{\mathrm{loc}}^{2,p}(\mathbb{R}^{N},\mathbb{R}^N)$ and ${\tilde{\bw}} + \nabla{\tilde \eta}  = {0}$ in $\Omega^{c}$, then $ {\tilde{\bw}} + \nabla{\tilde \eta} \in {W}_{0}^{2,p}(\Omega,\mathbb{R}^N)$ . Similarly, $\eta - \tilde \eta \in W^{2,p}_{0}(\Omega)$ due to $\eta \in W^{2,p}(\Omega), \tilde \eta \in W^{3,p}_{\mathrm{loc}}(\mathbb{R}^{3})$ and $\eta - \tilde \eta = 0$ in $\Omega^{c}$. Thus the proof is completed by \eqref{eqdecomposition}.
	\end{proof}
	
	By employing procedures analogous to those detailed above, the proof of Theorem \ref{thmdecompostion} can be established. It is crucial to recognize that a fundamental step in this proof is the extension of a function from $\Omega$ to $\mathbb{R}^N$. Subsequent lemmas can be derived using similar methods.
	\begin{theorem}
		\label{thmvdd}
		For the bounded Lipschitz domain $\Omega \subset \mathbb{R}^N$. The following decomposition holds
		\begin{equation}
			\label{decomposition35}
			\begin{aligned}
				L^{p}(\operatorname{curl}; \Omega,\mathbb{R}^N) &= W^{1,p}(\Omega,\mathbb{R}^N) + \nabla (W^{1,p}(\Omega)), \\
				L_0^{p}(\operatorname{curl}; \Omega,\mathbb{R}^N) &= W_0^{1,p}(\Omega,\mathbb{R}^N) + \nabla (W_0^{1,p}(\Omega)).
			\end{aligned}
		\end{equation}
	\end{theorem}
	
	\begin{theorem}
		\label{thmvddd}
		For the bounded Lipschitz domain $\Omega \subset \mathbb{R}^3$. The following decomposition holds
		\begin{equation}
			\label{decomposition36}
			\begin{aligned}
				W^{1,p}(\operatorname{curl}; \Omega,\mathbb{R}^3) &= W^{2,p}(\Omega,\mathbb{R}^3) + \nabla (W^{2,p}(\Omega)).
			\end{aligned}
		\end{equation}
	\end{theorem}
	
Extending functions with zero boundary conditions is straightforward and leads directly to the result in the second line of \eqref{decomposition35}. For the first row of \eqref{decomposition35}, the extension of functions in $L^{p}(\operatorname{curl}; \Omega,\mathbb{R}^N)$ adheres to similar procedures as those detailed in \cite[Lemma 2.2]{chen2000MatchingDC}. With respect to the decomposition result \eqref{decomposition36}, the methodologies described in \cite{hiptmair2002finite} are applicable to functions in $W^{1,p}(\operatorname{curl}; \Omega,\mathbb{R}^3)$.

	\section{Trace theorem for functions in $W^{m,p}(\Omega)$}

This section first presents some results of the trace theorem for functions in $W^{m,p}(\Omega)$, with $m \leq 3, 1 < p < \infty$, and Lipschitz domain $\Omega \subset \mathbb{R}^N$. Subsequently, it employs the generalized vector potential theorem in earlier sections to give an alternative proof strategy for the general trace theorem.
	
	Given $\phi \in W^{1, p}(\Omega)$, Gagliardo \cite{gagliardo1957caratterizzazioni} proves that the operator $\gamma_{0}(\phi) \coloneqq \phi_{\mid \Gamma}$ is linear and continuous from $W^{1, p}(\Omega)$ onto $W^{1-1 / p, p}(\Gamma)$ for $1 < p<\infty$ and has a continuous right inverse.  Let $\gamma_{1}(\phi) \coloneqq (\nabla \phi \cdot \boldsymbol{n})_{\mid \Gamma}$. The range of the trace operator $\left(\gamma_{0}, \gamma_{1}\right)$ with functions in ${W}^{2, p}(\Omega)$ for $1 < p<\infty$ has been solved as follows, see \cite{grisvard2011elliptic, geymonat2000existence, duran2001traces} for $N=2$ and  \cite{buffa2001traces} for $N=3$. 
	
	\begin{theorem}
		\label{tracek=11}
		The range of $\left(\gamma_{0}, \gamma_{1}\right)$ is the set of $\left(\phi_{0}, \phi_{1}\right) \in W^{1, p}(\Gamma) \times$ $L^{p}(\Gamma)$ such that
		$$
		\nabla_{\Gamma} \phi_{0}+\phi_{1} \boldsymbol{n} \in W^{1-1/p, p}(\Gamma,\mathbb{R}^N).
		$$
	\end{theorem}
	
	The results of Theorem \ref{tracek=11} hold for arbitrary $N$.
	Let
	$
	\gamma_{2}(\phi) \coloneqq  [ (\nabla^2 \phi \cdot \boldsymbol{n} ) \cdot \boldsymbol{n} ]_{\mid \Gamma} \in L^{p}(\Gamma). 
	$
	Geymonat \cite{geymonat2007trace} gives necessary conditions of the range of trace operator $\left(\gamma_{0}, \gamma_{1}, \gamma_{2}\right)$ for functions in $W^{3, p}(\Omega)$ for $1 < p<\infty$ and arbitrary $N$ as follows.
	\begin{theorem}
		\label{necessaryp=21}
		Let be $\left(\phi_{0}, \phi_{1}, \phi_{2}\right) \in W^{1, p}(\Gamma) \times L^{p}(\Gamma) \times L^{p}(\Gamma)$. Then $\left(\phi_{0}, \phi_{1}, \phi_{2}\right) \in \operatorname{range}\left(\gamma_{0}, \gamma_{1}, \gamma_{2}\right)$ if and only if 
		$$
		\bs \coloneqq \nabla_{\Gamma} \phi_{0}+\phi_{1} \boldsymbol{n} \in W^{1, p}(\Gamma,\mathbb{R}^N).
		$$
		and:
		$$
		\bS \coloneqq \nabla_{\Gamma}\bs + \sum_{i=1}^{N-1} \partial_{{\tau}_{i}}\bs \cdot \boldsymbol{n}  \left( \boldsymbol{n}  \otimes \boldsymbol{\tau}_{i} \right) + \phi_{2}   \left( \boldsymbol{n}  \otimes \boldsymbol{n} \right) \in W^{1-1 / p, p}\left(\Gamma, \mathbb{S} \right).
		$$
	\end{theorem}
	
	Sufficient proof for Theorem \ref{necessaryp=21}  is provided for $N=2$ in \cite{aibeche2023trace}. This paper establishes the sufficiency of Theorem \ref{necessaryp=21} for arbitrary $N$. To demonstrate this sufficiency, the following lemma is essential.
	
	\begin{lemma}
		\label{w1}
		For any $i=1,2,\cdots,N$, let 
		\begin{equation}
			\label{lemmadefw}
			\begin{aligned}
				\phi_{0} & = \bs \cdot \boldsymbol{e}_{i}, \\
				\phi_{1} &= \sum_{j=1}^{N-1}( \partial_{{\tau}_{j}} \bs \cdot \boldsymbol{n} ) \boldsymbol{\tau}_{j} \cdot \boldsymbol{e}_{i} + \phi_{2} \boldsymbol{n} \cdot \boldsymbol{e}_{i},
			\end{aligned}
		\end{equation}
		then there exists a $\omega_{i} \in W^{2,p}(\Omega)$ such that $\gamma_{0}(\omega_{i}) = \phi_{0}$ and $\gamma_{1}(\omega_{i}) = \phi_{1}$.
	\end{lemma}
	
	\begin{proof}
		Only the lemma needs to be proved for the case $i=1$.
		Recall the definition of $\bs$ and $\bS$ in Theorem \ref{necessaryp=21}, it follows from \eqref{lemmadefw} that
		$$
		\begin{aligned}
			\nabla_{\Gamma}\phi_{0} + \phi_{1} \boldsymbol{n} 
			&=  \nabla_{\Gamma} \bs  \cdot \boldsymbol{e}_{1} +  \gamma_{1}(\omega_{1}) \boldsymbol{n}   \\
			& = \left[ \nabla_{\Gamma}\bs + \sum_{j=1}^{N-1} \partial_{{\tau}_{j}}(\bs)\boldsymbol{n}  \left( \boldsymbol{n}  \otimes \boldsymbol{\tau}_{j} \right) + \phi_{2}   \left( \boldsymbol{n}  \otimes \boldsymbol{n} \right) \right]  \cdot \boldsymbol{e}_{1}\\ 
			& = \bS \cdot \boldsymbol{e}_{1}   \in W^{1-1 / p, p}(\Gamma, \mathbb{R}^N).
		\end{aligned}
		$$    
		By Theorem \ref{tracek=11}, which provides a sufficient condition for functions in $W^{2,p}(\Omega)$, the proof is completed.
	\end{proof}
	
	The demonstration of the sufficiency of Theorem \ref{necessaryp=21} proceeds as follows.
	\begin{proof}[Sufficient proof of Theorem \ref{necessaryp=21}]
		The proof can be divided into three steps. \\
		{\textbf{Step 1} constructs $\boldsymbol{v}=(\omega_{1},\omega_{2},\cdots, \omega_{N})- \nabla \theta \in W^{1,p}(\Omega,\mathbb{R}^N) $}.  
		By Theorem \ref{tracek=11}, there exists a function $\theta \in W^{2,p}(\Omega)$ satisfies $\gamma_{0}(\theta) = f_{0}$ and $\gamma_{1}(\theta) =f_{1}$. Additionally, one can construct $(\omega_{1},\omega_{2},\cdots, \omega_{N})  \in W^{2,p}(\Omega,\mathbb{R}^N)$ using the trace function $\bs$ due to Lemma \ref{w1}. \\
		\textbf{Step 2} decomposes $\boldsymbol{v} = \bw + \nabla \eta$. It can be verified that 
		$$\gamma_{0}( v_i ) = \gamma_{0}(\omega_i ) - \gamma_{0}( \nabla \theta )\cdot \boldsymbol e_{i} = \bs \cdot \boldsymbol e_{i} - (\nabla_{\Gamma} \phi_{0} + \phi_{1}\boldsymbol n) \cdot \boldsymbol e_{i}  = 0.$$
		Thus $\boldsymbol v \in W^{1,p}_{0}(\Omega,\mathbb{R}^N)$. Additionally, one can verify that 
		$$
		\begin{aligned}
			[\operatorname{curl} \boldsymbol v  ]_{ij}&= [\operatorname{curl} (\omega_{1},\omega_{2},\cdots, \omega_{N}) - \operatorname{curl} (\nabla \theta) ]_{ij} \\
			& =\frac{1}{2} \left(\frac{\partial \omega_j}{\partial x_i}-\frac{\partial \omega_i}{\partial x_j} \right) \in W^{1,p}(\Omega, \mathbb{T}) , \quad 1 \leq i, j \leq N.  \\
			[\gamma_{0} (\operatorname{curl} \boldsymbol v )]_{ij}& = \frac{1}{2} (\gamma_{0}(\nabla \omega_{j}) \cdot \boldsymbol{e}_i -\gamma_{0}(\nabla \omega_{i}) \cdot \boldsymbol{e}_j ) \\
			&= \bS \cdot \boldsymbol{e}_{j} \cdot\boldsymbol{e}_{i} - \bS \cdot \boldsymbol{e}_{i} \cdot\boldsymbol{e}_{j} =0, \quad 1 \leq i, j \leq N.
		\end{aligned}
		$$
		It follows from $\bS \in \mathbb{S}$ that $\operatorname{curl} \boldsymbol v \in W^{1,p}_{0}(\Omega,\mathbb{T})$. Theorem \ref{thmmain} implies that there exist $\bw \in {W}_{0}^{2,p}(\Omega,\mathbb{R}^N) $ and $\eta \in W^{2,p}_{0}(\Omega)$ satisfy
		\begin{equation}
			\label{exist}
			\boldsymbol{v} = \bw + \nabla \eta,  \text{ in } \Omega.
		\end{equation}
		Combining \eqref{exist} with $\boldsymbol v = (\omega_{1},\omega_{2},\cdots, \omega_{N})- \nabla \theta$ yields
		\begin{equation}
			\label{eqkey}
			\nabla( \eta +\theta) =  (\omega_{1},\omega_{2},\cdots, \omega_{N}) - \bw, \text{ in } \Omega.
		\end{equation}
		{\textbf{Step 3} verifies that $\phi = \eta + \theta$ satisfies the following condition} 
		\begin{equation}
			\label{eqverify}
			\begin{aligned}
				\gamma_{0}(\phi) &= \gamma_{0}(\eta)  + \gamma_{0}(\theta) = 0 + \phi_{0} = \phi_{0},\text{ on } \Gamma. \\
				\gamma_{1}(\phi)  &= \gamma_{1}(\eta)  + \gamma_{1}(\theta) = 0 + \phi_{1} = \phi_{1}, \text{ on } \Gamma.\\
				\gamma_{2}(\phi)  &= \nabla(\nabla \phi )  \cdot \boldsymbol{n} \cdot \boldsymbol{n} \\
				&= \nabla (\omega_{1},\omega_{2},\cdots, \omega_{N})  \cdot \boldsymbol{n} \cdot \boldsymbol{n} -  \nabla \bw \cdot \boldsymbol{n}\\
				& = \bS \cdot  \boldsymbol{n} \cdot \boldsymbol{n}  = \phi_{2},  \text{ on } \Gamma.  \\
			\end{aligned}
		\end{equation}
		{	The proof is completed by noting that $\phi \in W^{3,p}(\Omega)$ since \eqref{eqkey} and \eqref{eqverify}.}
	\end{proof}

	Let $q > 0$ be any positive integer. Introduce notations as follows
	\begin{equation*}
		\begin{aligned}
			(\boldsymbol{n}\otimes^q) & \coloneqq \overbrace{ \boldsymbol{n}\otimes \boldsymbol{n} \otimes ~ \cdots ~\otimes \boldsymbol{n} }^{q} ,  \\
			(\boldsymbol{n}\cdot^q) &\coloneqq  \overbrace{  \boldsymbol{n} \cdot \boldsymbol{n} \cdot ~\cdots~ \cdot \boldsymbol{n} }^{q} , \\
			\nabla_{\Gamma} (\cdot ) &\coloneqq \sum_{i=1}^{N-1} \boldsymbol{\tau}_i \otimes \partial_{\tau_i} (\cdot).
		\end{aligned}
	\end{equation*}
	Without causing ambiguity, $\mathbb{S}$ will be used to represent symmetric tensors in higher-order tensors, namely, if a tensor belongs to $\mathbb{S}$, swapping any two component indices leaves the tensor unchanged.
	Now it is time to give the generalized trace theorem.
	\begin{theorem}[Trace theorem for functions in $W^{m,p}(\Omega)$]
		\label{thmgeneraltrace}
		Let $\Omega$ be a bounded Lipschitz domain in $\mathbb{R}^N$. Given $(\phi_{0},\phi_{1},\dots,\phi_{m-1})\in W^{1,p}(\Gamma)\times L^p(\Gamma)\times \cdots \times L^{p}(\Gamma)$. Let $S_0 \coloneqq \phi_0$, 
		$$
		\bS_m \coloneqq \nabla_{\Gamma} \bS_{m-1} + \sum_{k=1}^{m-1} \sum_{i=1}^{N-1} (\boldsymbol{n} \otimes^k ) \otimes \boldsymbol{\tau}_i  \otimes \partial_{\tau_{i}} (\bS_{m-1}) \cdot (\boldsymbol{n} \cdot ^k)  +(\boldsymbol{n} \otimes^m) \phi_{m-1}, ~\forall m > 0.
		$$
		Then there exists $\phi \in W^{m,p}(\Omega)$ such that 
		$$\mathrm{Tr}_{m-1}\phi = (\gamma_{0}(\phi), \gamma_{1}(\phi), \dots, \gamma_{m-1}(\phi)) = (\phi_{0},\phi_{1},\dots,\phi_{m-1}),$$ 
		if and only if 
		$$
		\bS_q \in W^{1,p}(\Gamma,\mathbb{S}), 0 \leq q \leq m-2,  \quad \bS_{m-1} \in W^{1-1/p,p}(\Gamma,\mathbb{S} ).
		$$
	\end{theorem}
	
Theorem \ref{thmgeneraltrace} can be established using a recursive method. The primary procedures are akin to those used in Theorem \ref{necessaryp=21}, differing in that several critical results for the high-dimensional case are derived from Theorem \ref{thmdecompostion}.

	\bibliographystyle{plain}
	\bibliography{reference}

\end{document}